\documentclass[english,11pt,a4paper,dvipsnames]{article}
\ifx\pdfoutput\undefined\else
\pdfoutput=1
\fi 
\usepackage{amsmath,amsfonts,amssymb,amsthm}
\usepackage{graphicx,geometry,float,textpos,hyperref,microtype}
\newcommand{\orcid}[1]{\href{https://orcid.org/#1}{\includegraphics[height=1.8ex]{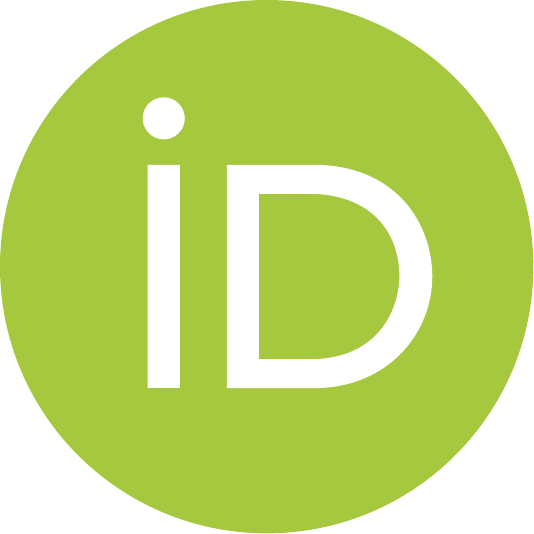}}}

\makeatletter
\def\blfootnote{\gdef\@thefnmark{}\@footnotetext}
\makeatother

\title{Faster Convolutions:\\Yates and Strassen Revisited%
\blfootnote{\rightskip=5.7cm
The research is funded by the European Union (ERC, CountHom, 101077083).
Views and opinions expressed are those of the author(s) only and do not necessarily reflect those of the European Union or the European Research Council Executive Agency. Neither the European Union nor the granting authority can be held responsible for them.
}
}
\author{
Cornelius Brand \orcid{0000-0002-1929-055X}\\
University of Regensburg
\and
Radu Curticapean \orcid{0000-0001-7201-9905}\\
\qquad University of Regensburg and IT University of Copenhagen \qquad
\and
Baitian Li \orcid{0009-0003-7408-9781}\\
Institute for Interdisciplinary Information Sciences, Tsinghua University
\and
Kevin Pratt \orcid{0000-0002-2923-0905}\\
Courant Institute of Mathematical Sciences, New York University
}
\date{}

\newcommand{\fconv}[1][f]{\mathbin{\circledast_{#1}}}
\newcommand{\rank}{\mathrm{rk}}

\usepackage[textwidth=2.5cm]{todonotes}

\newtheorem{theorem}{Theorem}[section]
\newtheorem{lemma}[theorem]{Lemma}

\theoremstyle{definition}

\theoremstyle{remark}
\newtheorem{remark}[theorem]{Remark}

\begin{document}

\maketitle

\begin{textblock}{5}(7.85, 7.35) \includegraphics[width=150px]{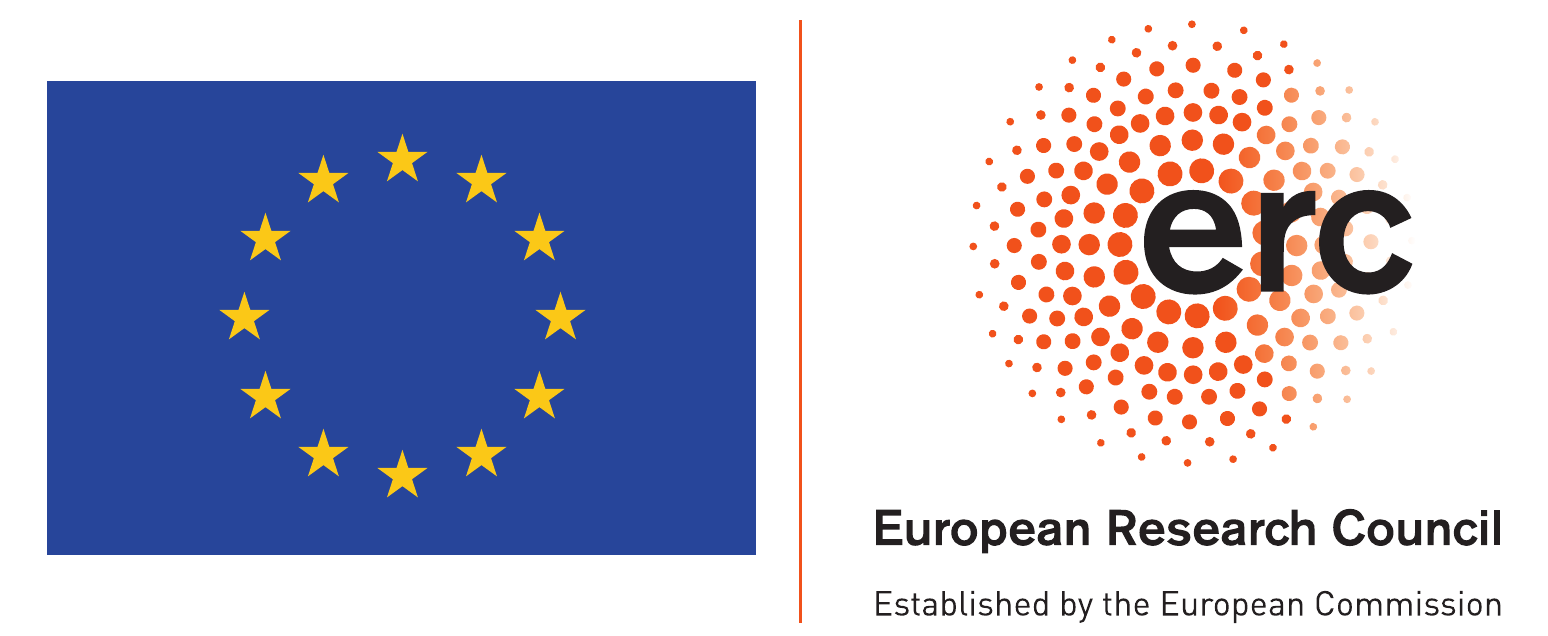} \end{textblock}

\begin{abstract}
Given two vectors $u,v \in \mathbb{Q}^D$ over a finite domain $D$ and a function $f : D\times D\to D$, the convolution problem asks to compute the vector $w \in \mathbb{Q}^D$ whose entries are defined by
\[
w(d) = \sum_{\substack{x,y \in D \\ f(x,y)=d}} u(x)v(y).
\]
In parameterized and exponential-time algorithms, convolutions on \emph{product domains} are particularly prominent:
Here, a finite domain $B$ and a function $h : B \times B \to B$ are fixed, and convolution is done over the product domain $D = B^k$, using the function $h^k :D \times D\to D$ that applies $h$ coordinate-wise to its input tuples.

We present a new perspective on product-domain convolutions through multilinear algebra. This viewpoint streamlines the presentation and analysis of existing algorithms, such as those by van Rooij \emph{et al.}~(ESA 2009). Moreover, using established results from the theory of fast matrix multiplication, we derive improved $O^\ast(|B|^{2\omega/3 \cdot k}) = O(|D|^{1.582})$ time algorithms, improving upon previous upper bounds by Esmer \emph{et al.}~(Algorithmica 86(1), 2024) of the form $c^k |B|^{2k}$ for $c < 1$. Using the setup described in this note, Strassen's asymptotic rank conjecture from algebraic complexity theory would imply quasi-linear $|D|^{1+o(1)}$ time algorithms. This conjecture has recently gained attention in the algorithms community. (Björklund-Kaski and Pratt, STOC 2024, Björklund \emph{et al.}, SODA 2025)

Our paper is intended as a self-contained exposition for an algorithms audience, and it includes all essential mathematical prerequisites with explicit coordinate-based notation. In particular, we assume no knowledge in abstract algebra.
\end{abstract}
\newpage

\section{Introduction}

\paragraph{Algorithms on Tree-Decompositions.}
Many $\mathsf{NP}$-hard graph problems admit $c^t \cdot n^{O(1)}$ time algorithms when a tree-decomposition of width $t$ is given along with an $n$-vertex input graph $G$. 
This includes finding a minimum dominating set, counting perfect matchings, and many other problems.
The algorithms for these problems process the bags of a nice tree-decomposition, proceeding from the leaves to the root, to construct partial solutions for the problem from previously computed partial solutions.
More specifically, for every bag $\beta$ in the tree-decomposition, the algorithms maintain some statistics (typically maximum size or count) about partial solutions induced by vertices from bags below $\beta$. These statistics are broken down by certain possible \emph{bag states} for $\beta$. In many problems, a bag state takes on the form of a tuple of \emph{vertex states}---one for each vertex in $\beta$. For example:
\begin{itemize}
\item The vertex states for the $3$-coloring problem are given by the three possible colors. 
\item For the dominating set problem, Telle and Proskurowski~\cite{TelleP97} gave a natural set of three vertex states. Namely, a vertex is either (1) part of the solution itself, or (2) it is dominated already by a vertex encountered down the tree, or (3) it will be dominated in some future bag further up the tree. 
\end{itemize}

\paragraph{Join Nodes.}

The crux with this approach then arises at \emph{join nodes}, where partial solutions of the children $\alpha,\alpha'$ of a bag $\beta$ need to be combined to partial solutions at $\beta$. This can become a bottleneck when vertex states for $v$ in the join node $\beta$ can arise from multiple possible vertex states for $v$ in the child bags $\alpha$ and $\alpha'$. Reconsidering our previous examples:
\begin{itemize}
    \item In the example of $3$-coloring, this problem does not arise, since it is clear that vertices appearing in the children should have the same color in the parent. 
    \item However, in the dominating set problem, a vertex might be dominated by a partial solution in either the left, the right, or both sub-trees of a node. Then, considering all different cases separately yields a running time of $O^\ast(5^t)$ instead of $O^\ast(3^t)$.
\end{itemize}

Alber \emph{et al.}~\cite{AlberBFKN02} and, subsequently, Van Rooij \emph{et al.}~\cite{RooijBR09} have shown how to turn Telle and Proskurowski's original set of three states into another set of three states with subtly modified, but overall equivalent semantics, eliminating redundancies when combining partial solutions. This allows us to solve the problem in $O^\ast(3^t)$ time, which is tight under SETH~\cite{LokshtanovMS18}. In particular, the state of ``not part of a solution, but dominated already,'' can be replaced by simply ``not part of the solution,'' and remaining agnostic as to whether the vertex is dominated. This creates redundant overlap with the remaining states, which then has to be handled by an adapted formula in the dynamic program, but overall leads to an improved running time. 

\paragraph{Convolutions.}
These modified semantics were designed carefully and cleverly, but in a very problem-specific manner. 
Van Rooij et al.~\cite{RooijBR09} and others more generally connected the processing of join nodes to so-called \emph{convolutions}, which are mathematical primitives that can precisely capture how combinations of child bag states yield parent bag states. As it turns out, for dominating sets as well as a host of related problems on tree decompositions, algorithmic progress is largely enabled by devising seemingly specialized convolution operations for specific encodings of the state set. This has sparked interest in computing such convolutions in more generality from the point of view of parameterized complexity~\cite{EsmerKMSW24}. 

\paragraph{Our Contributions.}
The careful modification of the semantics of the vertex state set and the efficient, ad-hoc implementation of accompanying convolutions for bag states make it seem like a sort of ``creative spark'' is needed to derive better algorithms for computation at join nodes. 
In this note, we wish to add two insights:
\begin{enumerate}
    \item We demonstrate that the optimal equivalent semantics for the state set and the respective convolution can be determined mechanically via \emph{low-rank tensor decompositions}, a well-known mathematical optimization problem.
    \item We show that the arising convolutions can generally be performed much more efficiently than is widely known in the algorithms community, through established results within algebraic complexity theory.
\end{enumerate}

\subsection{Convolution Problems}
The technical setup in general convolution problems is as follows.
We fix a finite domain $D$ and work with $D$-indexed vectors $v \in \mathbb Q^D$, i.e., functions $v : D\to \mathbb Q$.
In the algorithmic interpretation, $D$ is the set of possible bag states, and vectors $v \in \mathbb Q^D$ store the relevant statistics about partial solutions, broken down by bag state: The entry $v(d)$ for $d\in D$ stores some data about partial solutions with state $d$.  
We write $e_d \in \mathbb \{0,1\}^D$ for $d\in D$ to denote the canonical basis vector with $e_d(j)=1$ iff $j=d$.

Next, we fix a \emph{partial} function $f : D\times D \to D$.
This function captures how two children bag states compose to a parent bag state.
Given two $D$-indexed vectors $u,v\in \mathbb Q^D$, our aim then is to compute their convolution, which is the $D$-indexed vector
\begin{align} \label{eq:convdev}
u \fconv v = \sum_{x,y\in D}u(x)v(y) e_{f(x,y)},
\end{align}
where the sum on the right-hand side of course only ranges over those $x,y\in D$ such that $f(x,y)$ is defined.
To be explicit, the data stored for state $d\in D$ reads
\[
(u \fconv v)(d) = \sum_{\substack{x,y \in D \\ f(x,y)=d}}u(x)v(y).
\]
\begin{remark}
    We point out already here that Esmer \emph{et al.}~\cite{EsmerKMSW24} have similarly defined generalized $f$-convolutions, with the subtle difference of requiring $f$ to be a total function. Their results seem to rely on this restriction, but unfortunately, this requirement excludes most examples in the parameterized algorithms literature from consideration.
\end{remark} 

An algorithm using $O(|D|^2)$ arithmetic operations for the computation of \eqref{eq:convdev} is immediate, but near-linear algorithms are known for many important functions $f$.
Many of these cases exhibit a ``product structure,'' capturing that bag states are tuples of vertex states from a constant-sized set $B$:
Given a partial function $h : B\times B\to B$ and $k \in \mathbb N$, where $B$ and $h$ are fixed, we define $h^k :B^k\times B^k\to B^k$ by applying $h$ pointwise to tuples. That is, 
\begin{equation}
\label{eq:product-fn}
    h^k((x_1,\ldots,x_k),(y_1,\ldots,y_k)) = (h(x_1,y_1),\ldots,h(x_k,y_k)).    
\end{equation}
To be explicit, the partial function $h^k$ is undefined whenever one of the entries of the tuple on the right is undefined.
We say that $f=h^k$ exhibits product structure with base domain $D$ and base function $h$.

\subsection{Processing Join Nodes Using Convolutions}
To connect back to our original motivation, let us phrase a typical join operation in a tree decomposition using convolutions: Assume that every vertex in a bag can have a state from a set $B$ with $|B|=d$.
Then the possible states of a single bag are given by $B^{w}$, where $w$ is the size of the bag. For each such state $s\in B^w$, we typically store some number $T[s]$ in a table $T$. Now, let $f$ be a function such that, whenever a vertex $v$ occurs with state $x\in B$ in the left child, and with state $y \in B$ in the right child, it will occur with state $f(x,y) \in B$ in the parent join node. In this setup, the entries of the join node are then often given precisely by the convolution of the tables in the child nodes, i.e. $T[z] = \sum_{f^w(x,y)=z} T[x] \cdot T[y].$

We consider several examples to flesh out this abstract skeleton. For more background on dynamic programming on tree decompositions and problem definitions, we refer to any introduction to parameterized algorithms, such as~\cite{CyganFKLMPPS15}.
For a graphical representation of the convolutions, see Fig.~\ref{fig:tensors}.

\paragraph*{3-Colorings and the Trivial Convolution.} For the 3-coloring problem, we build a table $T_\beta$ at each bag $\beta$ such that $T_\beta$ contains a $0/1$-value for every assignment $\chi: \beta \rightarrow \{1,2,3\}$ of colors to bag vertices. The table entry at $\chi$ encodes whether there is a solution for the subgraph induced by all vertices in the bags below $\beta$ in the tree decomposition such that this solution agrees with $\chi$ on the vertex set $\beta$. 
In a join node with children $\alpha,\alpha'$, we have $T_\beta[\chi] = T_\alpha[\chi] \cdot T_{\alpha'}[\chi]$ for all $\chi$. The table $T_\beta$ can be computed trivially in time $O^\ast(3^t)$ from $T_\alpha$ and $T_{\alpha'}$. Explicitly, it corresponds to the convolution of $f^{|\beta|}$, where $f(x,x) = x$ for all $x \in D$, and $f(x,y)$ is undefined everywhere else.

\paragraph*{Perfect Matchings and Subset Convolution.} To count perfect matchings, we store states $1$ or $0$ for every vertex in $\beta$, encoding whether the vertex is already matched or still unmatched in the subgraph induced by all vertices below $\beta$. The entry $T_\beta[\mu]$ for $\mu: \beta \rightarrow \{0,1\}$ is the number of matchings that agree precisely with $\mu$ on $\beta$ and match all vertices in the bags in the tree below $\beta$ not themselves contained in $\beta$. Choosing the partial function $f$ with $f(x,y) = x \lor y$ for $(x,y)\neq (1,1)$, and $f(1,1)$ undefined, gives the correct convolution for counting perfect matchings, and it can trivially be implemented in time $O^\ast(3^t)$. 

This specific kind of convolution has risen to some prominence in parameterized and exact algorithms, under the name of \emph{subset convolutions}~\cite{BjorklundHKK07}, which we will expand on below.
In particular, $O^\ast(2^t)$ time algorithms are known.

\paragraph*{Dominating Sets and the Covering Product.} 
Dominating sets of size $k$ can be counted by using layered tables $T_\beta[\sigma,i]$ for $\sigma: \beta \rightarrow 
\{\text{in},\text{future},\text{past}\}$ and $0 \leq i \leq k$. The semantics of the states encode that a vertex in $\beta$ is itself part of the dominating set (in), it is not part of the dominating set and will only be dominated in a future bag above the bag $\beta$ (future), or it is not part of the dominating set and is already dominated in a past bag below the bag $\beta$ (past). With this setup,  $T_\beta[\sigma,i]$ will be filled as to contain the number of partial dominating sets of size $i$ that dominate all vertices in the graph induced by the vertices below $\beta$, except for those in state future. To express the combination of states as a convolution, define 
\begin{align*}
& f(\text{in},\text{in}) & =\qquad & \text{in},\\
& f(\text{past},\text{future}) = f(\text{future},\text{past}) = f(\text{past},\text{past}) & =\qquad & \text{past},\\
& f(\text{future},\text{future}) & =\qquad & \text{future},
\end{align*}
and leave $f$ undefined everywhere else. Then, let $\tau_{ij}[z] = \sum_{f^{\beta}(x,y) = z} T_\alpha[x,i] \cdot T_{\alpha'}[y,j]$ and set $T_\beta[\sigma,\ell] = \sum_{ij} \tau_{ij}$, ranging over all pairs with $i+j+|\sigma^{-1}(\text{in})|=\ell$.
The convolutions can be naively implemented using $O^\ast(5^t)$ arithmetic operations. 

This example contains an interesting sub-case: If we restrict the domain of $f$ to $\{\text{future},\text{past}\}$, then $f$ gives rise to the so-called \emph{covering product}, which is closely related to the subset convolution mentioned above. This restriction lies at the heart of known improvements to $O^\ast(3^t)$ time~\cite{RooijBR09}.

\paragraph*{Longest Paths and XOR-Convolution.} 
For the sake of completeness, we mention that convolutions also occur in parameterized algorithms that do not exploit tree-decompositions.
For instance, let $D = \{0,1\}$ and let $f = \oplus$ be the Boolean XOR. The $n$-th power of the generalized $f$-convolution has made a prominent appearance in fixed-parameter algorithms for finding paths of length $k$ (parameterized by $k$)~\cite{Koutis08,Williams09,KoutisW16}. In these applications, the field $\mathbb Q$ is usually replaced by the field with $2^q$ elements, but this is immaterial for our purposes.
Spelled out, two vectors $u,v \in \mathbb Q^{\mathbb Z_2^n}$ are convolved through
\[
u \fconv[f] v = \sum_{x,y \in \mathbb Z_2^n} u(x)v(y) e_{x \oplus y}.
\]
This example is the simplest in a whole class of important convolutions, namely those where $D$ is a group and $f$ is the group operation. In the commutative case, this leads to multidimensional DFTs, while the general case is governed by the representation theory of finite groups~\cite{Umans19}.

\begin{figure}\centering
\includegraphics{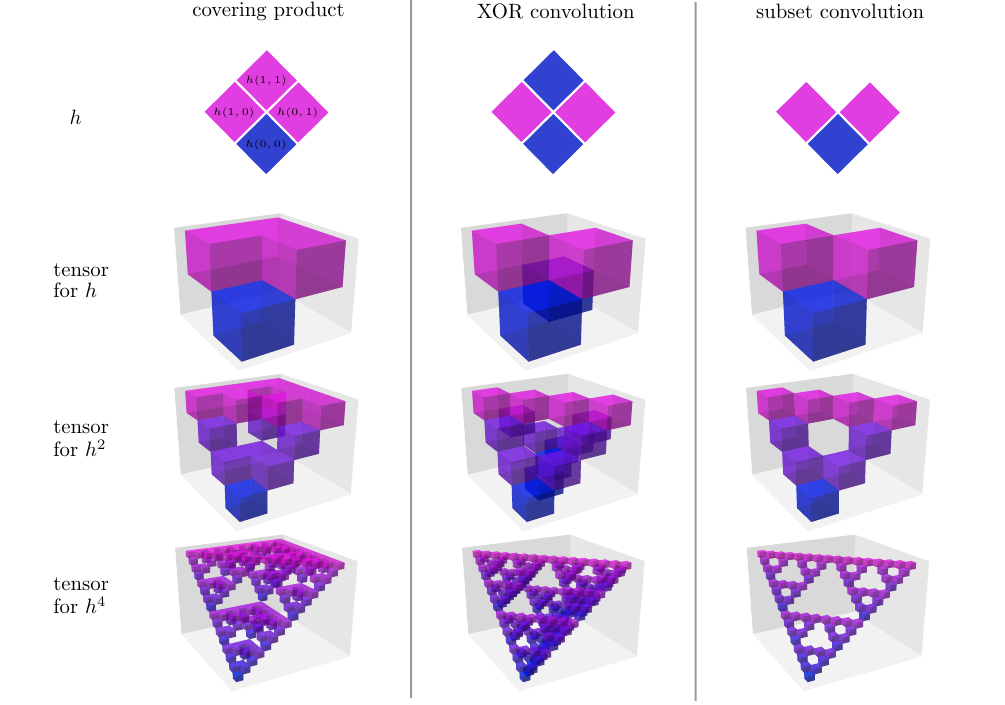}
\caption{Fix the base domain $B=\{0,1\}$ and one of the three base functions $h: B\times B\to B$ shown in the first row, where blue, pink and white squares correspond to the function values $0$, $1$, and undefined, respectively.
The covering product, XOR and subset convolutions are powers $f = h^k$ of the respective base function that induce vector-valued functions $t:D\times D\to \mathbb \{0,1\}^D$ for $D=B^k$.
The functions $t$ are depicted as tensors, i.e., $3$-dimensional arrays with $0/1$ entries, where $t(x,y)$ is stored in the vertical column $(x,y,\star)$, and a block at position $(x,y,z)$ indicates that the $z$-th entry of $t(x,y)$ equals $1$. (For aesthetic purposes, the block color encodes the Hamming weight of $z$; it does \emph{not} encode the array value, which is $0$ or $1$.)
}
\label{fig:tensors}
\end{figure}

\section{Algorithmic Results}
We now describe how to reformulate convolutions in terms of linear algebra, and then apply two algorithmic results that are well-known in the algebraic complexity literature. 
Our aim is to give an accessible, from-scratch account.

\subsection{Framework: Convolutions as Bilinear Maps}
\label{sec:framework}
For every function $f$, the convolution is a map of type $u \fconv v \,:\,\mathbb Q^D\times \mathbb Q^D\to \mathbb Q^D$. We observe that it is \emph{bilinear}, i.e., linear in both arguments: 
\begin{align*}
(u+u') \fconv v & =u \fconv v + u' \fconv v\\
u\fconv(v+v') & =u\fconv v+u\fconv v'.
\end{align*}
By modifying \eqref{eq:convdev} slightly, \emph{every} bilinear map $f:\mathbb Q^D\times \mathbb Q^D\to \mathbb Q^D$ can be viewed as a generalized convolution.
This generalized perspective may appear like a mathematical curiosity, but it is actually crucial for us.
Formally, we fix an arbitrary function $t : D \times D \to \mathbb Q^D$ and define
\begin{align} \label{eq:def-bilin}
u \fconv[t] v = \sum_{x,y\in D}u(x)v(y) t(x,y).
\end{align}

This recovers \eqref{eq:convdev} by setting $t(x,y)=e_{f(x,y)}$ whenever $f(x,y)$ is defined, and $t(x,y) = 0$ otherwise.
We will later require more general vectors than $e_d$ as outputs for $t$.

A natural notion of product structure can also be defined in this setting: We say that the function $t : D \times D \to \mathbb Q^D$ admits product structure if $D=B^k$ for some fixed domain $B$ and there is a fixed $b:B\times B\to\mathbb Q^B$ such that the vector $t((x_1,\ldots,x_k),(y_1,\ldots,y_k))$ has its entry at index $(z_1,\ldots,z_k) \in D^k$ given by 
\begin{equation}
b(x_1,y_1,z_1) \cdot  \ldots \cdot b(x_k,y_k,z_k).
\end{equation} 
The vector in $\mathbb Q^D$ defined this way is called the Kronecker product $b(x_1,y_1)\otimes \ldots \otimes b(x_k,y_k) \in \mathbb Q^D$ of the vectors $b(x_1,y_1),\ldots,b(x_k,y_k).$

\begin{remark}
    The data $(t(x,y)_k)_{x,y,k\in D} \in \mathbb Q^{D^3}$ often are referred to as the \emph{tensor} associated with a bilinear map. They are completely determined by the bilinear map; conversely, every such tensor defines a bilinear map. In the algebraic complexity literature, often, little reference is made to the bilinear map, and results are purely formulated in tensor language. We refrain from this impulse in the following for the sake of concreteness and interpretability.
\end{remark}

\subsection{Rank}
The linear-algebraic interpretation of convolutions by means of $t : D \times D \to \mathbb Q^D$ allows us to decompose $t$ into sums of simpler functions, leading to the notion of \emph{rank}:
A function $t : D\times D\to \mathbb Q^D$ has rank $1$ if $t(x,y)=a(x)b(y)c$ for some vectors $a,b,c \in \mathbb Q^D$.
The convolution $\fconv[t]$ for functions $t$ of rank $1$ simplifies drastically:
\begin{align}
u\fconv[t]v\  & =\ \sum_{x,y\in D}u(x)v(y)\underbrace{a(x)b(y)c}_{t(x,y)} =\ \underbrace{\left(\sum_{x\in D}u(x)a(x)\right)}_{\langle u,a\rangle}\underbrace{\left(\sum_{y\in D}v(y)b(y)\right)}_{\langle v,b\rangle}c.\label{eq:conv-rk1}
\end{align}
Using this formula, we can directly evaluate $u\fconv[t]v$ with $O(|D|)$ operations when $t$ has rank $1$.
More generally, we can evaluate convolutions $u\fconv[t]v$ with few operations when $t$ is a linear combination of a small number of functions of rank $1$:
The rank $\rank(t)$ of $t : D\times D \to \mathbb Q^D$ is the minimum $r\in \mathbb N$ such that $t = t_1+\ldots +t_r$ pointwise with each $t_i$ of rank $1$.
\begin{remark}\label{rem:trivialgo}
Using \eqref{eq:conv-rk1} and linearity, we can immediately compute $u \fconv[t]v$ with $O(\rank(t)\cdot |D|)$ operations.
This will however turn out not to be optimal in important cases.
\end{remark}

For arbitrary $f$, we have $\rank(t)\leq |D|^2$, since $t = \sum_{i,j \in D} t_{i,j}$ with the rank-$1$ functions $t_{i,j} = e_i (x)e_j(y)t(i,j)$ for $i,j\in D$.
Several important functions $t$ have strictly lower rank than $|D|^2$. 
This fact enables speed-ups, e.g., over the dynamic programming algorithms for perfect matchings and dominating sets outlined above.
We consider some examples in the following, with a particular focus on the interaction among product structure and rank.

\subsubsection*{Covering Product: Rank and Product Structure}
Letting $D = \{0,1\}^{k}$ and defining $f(x,y) = x \lor y$ component-wise for $x,y \in \{0,1\}^k$, the corresponding convolution operation
\[
u \fconv v
= \sum_{x,y \in \{0,1\}^k} u(x) v(y) e_{x \lor y}
\]
is the \emph{covering product} on the power set of $[k]$, represented by characteristic vectors. A naive implementation of this convolution uses $\Omega(3^k)$ arithmetic operations. It has been known since the 60s~\cite{Solomon1967} that this is not optimal, although the structural insights enabling this realization were only algorithmically exploited some 40 years later~\cite{BjorklundHKK07}.

\paragraph{Low Rank Lifts to Higher Powers.}
The decisive observation for the covering product can already be made for $k=1$, where $f = \lor$ and the corresponding function $t: D\times D\rightarrow \mathbb Q^D, (x,y) \to e_{x \lor y}$ has rank at most $3$, because its support contains only $3$ pairs. However, the actual rank of $t$ is $2$ rather than $3$: 
We have  $t = t_1 + t_2$ for $t_1(x,y) = a_1(x)b_1(y)c_1$ and $t_2(x,y) = a_2(x)b_2(y)c_2$, where
\begin{align*}
   & a_1(x) = b_1(y) = 1, & \text{ for } x,y \in \{0,1\}, \\
   & a_2(0) = b_2(0) = 1,\ a_2(1) = b_2(1) = 0, & \text{ and } \\
   & c_1 = e_1, ~ c_2 = e_0 - e_1.
\end{align*}

One crucial feature of the rank-based approach is that, whenever dealing with product-structured $t$ with base function $b$, proving rank upper bounds for fixed $t$ already implies upper bounds for all powers $b^k$ of $b$.
In other words, low-rank decompositions for $b$ ``lift'' to powers $b^k$.
\begin{lemma}
    If $t = b^k$, then we have
    \begin{equation}\label{eq:kronecker-rank-ub}
        \rank(t) \leq (\rank(b))^k.
    \end{equation}    
\end{lemma}
\begin{proof}
For brevity, set $r = \rank(b)$, and denote with $\vec x, \vec y,\vec z$ elements of $D^k$. Then~\eqref{eq:kronecker-rank-ub} follows because $b(x,y) = \sum_{i=1}^{r} a_i(x) b_i(y) c_i$ implies 
\begin{align*}
& b^k(\vec x, \vec y, \vec z) =\  b(x_1,y_1,z_1)\cdots b(x_k,y_k,z_k) \\ 
= \ & \sum_{\iota: [k] \rightarrow [r]} a_{\iota(1)}(x_1)b_{\iota(1)}(y_1) c_{\iota(1)}(z_1) \cdots a_{\iota(k)}(x_k)b_{\iota(k)}(y_k) c_{\iota(k)}(z_k).
\end{align*}
Then, defining for $\iota: [k] \rightarrow [r]$ the maps
\begin{equation}\label{eq:abc}
A_{\iota}(\vec x)=\prod_{j=1}^{k} a_{\iota({j})}(x_{j}),\ B_{\iota}(\vec y)=\prod_{j=1}^{k} b_{\iota({j})}(y_{j}),\ \text{ and } C_{\iota}(\vec z)=\prod_{j=1}^{k}c_{\iota({j})}(z_{j})
\end{equation}
allows us to obtain a decomposition of $b^k$ via 
\begin{align} \label{eq:kronecker-rank-decomp}
b^k((x_1,\ldots,x_k),(y_1,\ldots,y_k)) = \sum_{\iota: [k] \rightarrow [r] } A_{\iota}(x_1,\ldots,x_k) B_{\iota}(y_1,\ldots,y_k) C_{\iota},
\end{align}
proving the inequality \eqref{eq:kronecker-rank-ub}.
\end{proof}

Hence, \eqref{eq:kronecker-rank-ub} allows to deduce from the case $k=1$ that the rank of general covering products with $k > 1$ is bounded by $2^k$. By Remark~\ref{rem:trivialgo}, this would only give an algorithm for general covering products running in time $O^\ast(4^k)$. However, an $O^\ast(2^k)$ time bound is implied by Yates' algorithm, which we state as Theorem~\ref{thm:yates} and describe in Section~\ref{thm:yates}. 

Moreover, this decomposition can be used to show that also the function $f$ relevant for counting dominating sets on bounded-treewidth graphs described above can be computed using $O^\ast(3^t)$ time,
as was first shown by Van Rooij \emph{et al.}~\cite{RooijBR09}. Importantly, this fact is available without any semantic insights into the definition of the state set, and can be derived in an entirely mechanic manner, simply by deriving the best possible rank decomposition for the corresponding functions on vertex states and lifting them to products of bag states. 

\begin{remark}
    For completeness, note that for the rank of the XOR-convolution on $n$-bit strings defined above, the function $t$ associated with the corresponding $f$ has rank $2^n$. This is an essential prerequisite for the Fast Fourier Transform.
Again, considering the case $n=1$ suffices to derive higher powers. For this, set $t(0,1) = t(1,0) = e_1$ and $t(1,1) = t(0,0) = e_0$.
Picking 
\begin{align*}
    & a_1(x) = b_1(x) = c_1(x) = (-1)^x, \\ 
    & a_2(x) = b_2(x) = c_2(x) = \phantom{(-}1 
\end{align*}
for $x\in \{0,1\}$ yields $t = t_1 + t_2$,
where again $t_i(x,y) = a_i(x)b_i(y)c_i$, $i=1,2$.
\end{remark}

\paragraph{Lifting is Not Always Enough.}
In contrast to the covering product, for the closely related \emph{subset convolution}, where we set $f(1,1)$ to be undefined instead of $1$, no useful upper bound is available in the $k=1$ case. Instead, the best rank upper bound for $k=1$ is the trivial bound of
\begin{align}
    \rank(t) \leq 3,
\end{align}
where $t: \{0,1\} \times \{0,1\} \rightarrow \mathbb Q^{\{0,1\}}$ is the corresponding function for subset convolution. 

However, despite this worse bound for the simplest case, it is still true that the rank of subset convolution for general $k$ can be bounded by $O^\ast(2^k)$, which shows that \eqref{eq:kronecker-rank-ub} is not tight in general. The improved upper bound can be observed by expressing subset convolution as a polynomial-sized sum of covering products of vectors with fixed support sizes (which was termed ``ranked subset convolution'' by Björklund et al.~\cite{BjorklundHKK07}), each of which has rank at most $2^k.$ The best known upper bound is $O(2^k k),$ which can be achieved by polynomial interpolation of so-called border decompositions of the function in question, cf.~\cite{BjorklundK24}. This allows to implement operations at join nodes in the outlined algorithm for counting perfect matchings in time $O^\ast(2^t).$

In the next section, we will see in Theorem~\ref{thm:strassen} that product structure can generally be used in ways that go beyond lifting low-rank decompositions of the fixed base functions.

\subsection{Rank and Computation}
Recall that we can compute $u \fconv[t]v$ with $O(\rank(t)\cdot |D|)$ operations by Remark~\ref{rem:trivialgo}.
If $t$ admits product structure, faster algorithms are known. In particular, Yates' algorithm~\cite{yates} accomplishes the following; we give a proof in Section~\ref{sec:yates}:
\begin{theorem} \label{thm:yates}
Let $t : D\times D \to \mathbb Q^D$ have product structure with base function $h : B\times B\to \mathbb Q^B$ and $D=B^k$ for $k\in \mathbb N$. 
Let $h = h_1+\ldots+h_r$ with each $h_i$ of rank $1$ be given as input, where we assert that $|B| \in O(1)$ and $|B|\leq r \leq |B|^2$.
Then there is an $O(r^{k} k)$ time algorithm to compute $u \fconv[t]v \in \mathbb Q^D$ on input $u,v\in \mathbb Q^D$.
\end{theorem}
Generally, with $\omega < 2.372$ denoting the exponent of matrix multiplication, the following surprising theorem follows from known facts in algebraic complexity theory~\cite{Strassen1988}; we give a self-contained proof in Section~\ref{sec:strassen}.
\begin{theorem} \label{thm:strassen}
    If $t:D^k\times D^k \to \mathbb Q^{D^k}$ admits product structure, then $\fconv[t]$ can be computed with $O^\ast(|D|^{2 \omega /3\cdot k}) = O(|D|^{1.582k})$ operations.
    Both $|D|$ and $k$ are considered part of the input.
\end{theorem}
Moreover, the \emph{asymptotic rank conjecture}~\cite{Strassen1994}, a significant strengthening of the conjecture $\omega = 2$, implies optimal exponential running time of all product-type convolutions.
The conjecture essentially states that every function $t:D^k\times D^k \to \mathbb Q^{D^k}$ that admits product structure has rank bounded by $O(|D|^{k+o(k)})$;
the low-rank decompositions asserted by this conjecture generally need not have the product form described in the proof of \eqref{eq:kronecker-rank-ub}. 

This conjecture yields the following implication, by invoking Theorem~\ref{thm:yates} on arbitrarily large constant powers of the base function underlying $t$.
\begin{theorem} \label{thm:conv-asymprank}
    If $t:D^k\times D^k \to \mathbb Q^{D^k}$ admits product structure and the asymptotic rank conjecture holds, then $\fconv[t]$ can be computed with $O(|D|^{k + o(k)})$ operations.
\end{theorem}
The asymptotic rank conjecture is in fact even \emph{equivalent} (within an algebraic model of computation) to the claim that all convolutions can be computed in the stated running time, and we refer to the recent literature connecting this conjecture to the theory of fine-grained algorithms for further background~\cite{BjorklundK24,BjorklundCHKP25,Pratt24}.

\section{Proof of Theorem~\ref{thm:strassen}: Bilinear Split-and-List}
\label{sec:strassen}
We present Strassen's~\cite{Strassen1988} proof of Theorem~\ref{thm:strassen},
adopting a concrete presentation of the more abstract original that sacrifices some generality.
To recall, given a partial function $f: D \times D \rightarrow D$, its third power $f^3: D^3 \times D^3 \rightarrow D^3$ is given by 
\[
f^3(i_1,i_2,i_3;j_1,j_2,j_3) := f^3((i_1,i_2,i_3),(j_1,j_2,j_3))= (f(i_1,j_1),f(i_2,j_2),f(i_3,j_3))
\]
and the $f^3$-convolution of $u,v \in \mathbb Q^{D^3}$ reads 
\[
u \fconv[f^3] v = \sum_{\substack{i_1,i_2,i_3 \in D \\ j_1,j_2,j_3 \in D}} u_{i_1,i_2,i_3} v_{j_1,j_2,j_3} e_{f(i_1,j_1), f(i_2,j_2), f(i_3,j_3)},
\]
again only summing over tuples where $f$ is defined.
Strassen's construction relies on \emph{matrix multiplication}, specifically, of $D^2 \times D^2$ matrices.
For $U, V \in \mathbb Q^{D^2 \times D^2}$, we interpret the entries of $U$ and $V$ as a flattened vector of length $D^4$.
This allows us to understand matrix multiplication as a $t$-convolution, where 
\[
t: D^4 \times D^4 \rightarrow \mathbb Q^{D^4},\ 
t(a,b,i,j; i,j,c,d) = e_{a,b,c,d},
\]
and $t$ sends any other pair of 4-tuples not in this form to $0$. 
It follows that
\[
U \fconv[t] V = U\cdot V,
\]
interpreting $U,V$ as vectors of length $D^4$ on the left, and as matrices of size $D^2\times D^2$ on the right. Strassen discovered that convolution on $D^3$ by $f^3$ ``fits into'' $D^2 \times D^2$ matrix multiplication:
\begin{lemma} \label{lem:fastmult}
    Let $u,v \in \mathbb Q^{D^3}$. Then, one can compute matrices $U,V \in \mathbb Q^{D^2 \times D^2}$ as well as a linear mapping $\phi: \mathbb Q^{D^4} \rightarrow \mathbb Q^{D^3}$ such that 
    \[
        u \fconv[f^3] v = \phi(U \cdot V).
    \]
    Moreover, the computation of $U$ and $V$ as well as the evaluation of $\phi$ can be performed with $O(|D|^4)$ arithmetic operations.
\end{lemma}
\begin{proof}
    We define $U_{a,b;i,j} = \sum_{\ell : f(\ell,j) = b} u_{a,i,\ell}$ and $V_{i,j;c,d} = \sum_{k : f(i,k) = d} v_{c,k,j}$.
    Then,
    \begin{align*}
    U \fconv[t] V  & =\sum_{a,b;c,d} \sum_{i,j} U_{a,b,i,j} \cdot V_{i,j,c,d} e_{a,b;c,d}  \\
    & =\sum_{a,\ldots,j} \, \sum_{\substack{\ell : f(\ell,j) = b, \\ k: f(i,k) = d }} u_{a,i,\ell} v_{c,k,j} e_{a,b;c,d}  \\
    & =\sum_{a,c,i,j,k,\ell} u_{a,i,\ell} v_{c,k,j} e_{a,f(\ell,j);c,f(i,k)}.
    \end{align*}
    Now, setting $\phi(e_{a,b;c,d}) = e_{f(a,c),b,d}$ (or $0$ if $f(a,c)$ is undefined), we find that 
    \[
    \phi(U \fconv[t] V) = \sum_{a,c,i,j,k,\ell} u_{a,i,\ell} v_{c,k,j} e_{f(a,c),f(\ell,j),f(i,k)}.
    \]
    When building $U$, the set of summand indices is $\{(a,b = f(\ell,j),i,j,\ell) : a,\ldots,\ell \in D\}$,
    implying that only $O(|D|^4)$ operations are performed across all entries of $U$.
    The same argument applies to $V$. Finally, seeing that $\phi$ maps unit vectors to unit vectors, also this mapping can be performed in the claimed running time. This completes the construction.
\end{proof}
While the choice of entries in $U$ and $V$ may appear \emph{ad hoc}, they can be derived from Strassen's more abstract algebraic approach.
A more algorithmically minded derivation stems from the ``split \& list'' method common in fine-grained complexity: The entries of $U$ correspond, in some sense, to partial solutions on the left, while $V$ generates partial solutions on the right. These partial solutions can then be correctly combined by a matrix product. With this setup, the proof of Theorem~\ref{thm:strassen} is immediate:
\begin{proof}[Proof of Theorem~\ref{thm:strassen}]
    Let $n = \lceil k/3 \rceil$. Then, we consider the generalized $f$-convolution as computing a mapping $(D^n)^3 \times (D^n)^3 \rightarrow (D^n)^3$, padding with $0$ as needed.
    We black-box the $n$-fold internal product structure, and consider this an arbitrary convolution with respect to a function of type $B^3 \times B^3 \rightarrow B^3$, where $B = D^n$.
    
    By Lemma~\ref{lem:fastmult}, this can be expressed as the projection of a matrix multiplication of matrices of size $(|D|^n)^2$ by $(|D|^n)^2$,
    and this expression can be evaluated in time $O(|D|^{2\omega n} + |D|^{2n}) = O(|D|^{2\omega n}).$
    The claim then follows from the choice of $n$ for growing $k$ and $|D|$.
\end{proof}
Clearly, all operations performed in the proof of Theorem~\ref{thm:strassen} only blow up the bit-length of the input numbers encoded in binary by a constant factor. Hence, instead of counting arithmetic operations, we might also count bit operations.
\begin{remark}
  Strassen's result only makes implicit use of the product structure of $D^k$. In fact, the embedding into matrix multiplication is possible for every third power of a function, irrespective of the absence or presence of product structure of the base function.
  Asymptotically, requiring divisibility by three is immaterial. 
  (Ultimately, when then employing fast matrix multiplication algorithms, the product structure of matrix multiplication is crucially used in these algorithms, albeit in a black-box manner that has nothing to do with the embedding we constructed.)
\end{remark}

\section{Proof of Theorem~\ref{thm:yates}: Yates' algorithm}
\label{sec:yates}
We now give a concrete account of a proof of Yates' algorithm, aligned with the language of this note.
Many proofs are available in the literature, although many are stated in the more concise, but abstract language of tensor algebra (cf.~\cite{BjorklundK24}).

\begin{proof}[Proof of Theorem~\ref{thm:yates}]
Let $t : D\times D \to \mathbb Q^D$ have product structure with base function $h : B\times B\to \mathbb Q^B$ and $D=B^k$ for $k \in \mathbb N$. 
Let $h = h_1+\ldots+h_r$ with each $h_i$ of rank $1$.
For our analysis, we assume that $|B| \leq r\leq |B|^2$, which can be achieved by adding dummy terms to $h$ or by choosing a trivial decomposition of rank $|B|^2$.
For $i\in [r]$, let $a_i,b_i,c_i \in \mathbb Q^B$ witness that $h_i$ has rank $1$, i.e., 
\[
h_i(x,y) = a_i(x)b_i(y)c_i\quad \text{for all }x,y\in B.
\]
For $\kappa \in [r]^k$, we define $a_\kappa,b_\kappa,c_\kappa \in \mathbb Q^D$ using the Kronecker product from Section~\ref{sec:framework} as 
\begin{align*}
a_{\kappa} & =a_{\kappa(1)}\otimes\ldots\otimes a_{\kappa(k)},\\
b_{\kappa} & =b_{\kappa(1)}\otimes\ldots\otimes b_{\kappa(k)},\\
c_{\kappa} & =c_{\kappa(1)}\otimes\ldots\otimes c_{\kappa(k)}.
\end{align*}
This is essentially a more concise way of writing~\eqref{eq:abc}.
We then have
\begin{equation}\label{eq:conv-yates}
u \fconv[t]v 
= \sum_{\kappa} \langle u, a_\kappa\rangle \langle v, b_\kappa\rangle c_\kappa.
\end{equation}

In the remainder, we show how to compute the above sum in the claimed running time. 
For this, we proceed in two stages: 
In the first stage, we compute all inner products $\langle u, a_\kappa\rangle$ and $\langle v, b_\kappa\rangle$ by an inductive process. In the second stage, we group the terms in \eqref{eq:conv-yates} by a similar inductive process to evaluate the sum efficiently.
To describe this inductive process, it helps to imagine the intermediate results organized in a full $r$-ary tree $T$ with layers $0,\ldots,k$, where the nodes in layer $i$ are indexed by $i$-tuples $\iota\in[r]^{i}$.

In the first stage, we compute two vectors $u_\iota$ and $v_\iota$ of dimension $|B|^{k-i}$ for each node $\iota$, proceeding inductively from layer $0$ to $k$. 
The vectors at layer $i$ will be indexed by tuples $(x_i,\ldots ,x_{k-1}) \in B^{k-i}$, and the vectors at the final layer $k$ will ultimately be the scalars $\langle u, a_\kappa\rangle$ and $\langle v, b_\kappa\rangle$ for $\kappa \in [r]^k$. 
\begin{itemize}
    \item At the root $\epsilon$ of $T$, we set $u_\epsilon=u$ and $v_\epsilon = v$.
    \item Given a node $\iota \in [r]^i$ at layer $i<k$ and $u_\iota, v_\iota$, and given the index $t\in [r]$ of a child node $\iota t$, we compute $u_{\iota t}, v_{\iota t}$ at all entries $(x_{i+1},\ldots,x_{k-1}) \in B^{k-i-1}$ via
        \begin{align*}
        u_{\iota t}(x_{i+1},\ldots,x_{k-1}) & = \sum_{x\in B}u_{\iota}(x,x_{i+1},\ldots,x_{k-1})a_{t}(x),\\
        v_{\iota t}(x_{i+1},\ldots,x_{k-1}) & = \sum_{x\in B}v_{\iota}(x,x_{i+1},\ldots,x_{k-1})b_{t}(x).
        \end{align*}
        This computation clearly requires $O(|B|^{k-i})$ operations and can be interpreted as a ``pointwise inner product'' $u_{\iota t}(x_{i+1},\ldots,x_{k-1}) = \langle u_\iota (\star,x_{i+1},\ldots,x_{k-1}),a_t\rangle$, likewise for $v_{\iota t}$. 
        At layer $k$, we interpret $(x_k\ldots x_{k-1})$ as empty.
\end{itemize}

To analyze the overall running time for this stage, note that layer $i$ contains $r^i$ nodes, so the total time spent in layer $i$ is $O(|B|^{k-i}r^i)= O(|B|^k\cdot (r /|B|)^i)$.
The total time over all layers is a geometric series that amounts to $O(|B|^k\cdot (r /|B|)^k)=O(r^k)$ time.

In the second stage, we traverse the tree from the leaves to the root, i.e., from layer $k$ to $0$, in order to compute the sum over $r^k$ vectors given by \eqref{eq:conv-yates}:
For all $\iota \in [r]^i$ in layer $i$, we compute a 
vector $f_{\iota}$ of dimension $k-i$ such that 
\[
c_{\iota(1)}  \otimes \ldots \otimes c_{\iota(i)} \otimes f_{\iota}
= \sum_{\iota \theta \text{ leaf below }\iota} \langle u, a_{\iota\theta}\rangle \langle v, b_{\iota\theta}\rangle  c_{\iota\theta}.
\] 
For leaves $\kappa$, we set $f_\kappa = \langle u, a_{\kappa}\rangle \langle v, b_{\kappa}\rangle$.
For an internal node $\iota$ at level $i<k$, given the vectors $f_{\iota t}$ for the children $\iota t$ of $\iota$, with $t\in [r]$, we compute $f_{\iota}$ via $f_{\iota} = \sum_{t=1}^r c_{\iota(i)} \otimes  f_{\iota t}$.
As in the analysis for the first stage, the resulting geometric series obtained from counting all operations shows that we can also finish this stage in $O(r^k k)$ time. The correctness of the above computation follows directly from the definitions.
\end{proof}
    As in the proof of Theorem~\ref{thm:strassen}, note that the bit-length of the numbers in Yates' algorithm only grows polynomially in $k$ and $r$, meaning that the bound on the number of arithmetic operations is essentially a bound on the number of bit operations.

\bibliography{refs}

\end{document}